\documentclass[11pt,a4paper]{article}

\usepackage{amsmath}
\usepackage{amsfonts}
\usepackage{amsthm}

\usepackage{graphicx}

\newtheorem{Theorem}{Theorem}
\newtheorem{Corollary}{Corollary}

\newtheorem{Proposition}{Proposition}
\newtheorem{Lemma}{Lemma}
\newtheorem{Remark}{Remark}

\newcommand{\ds}{{\mathrm{d}}s}

\newcommand{\dy}{{\mathrm{d}}y}
\newcommand{\Id}{{\mathbf{1}}}
\newcommand{\dom}{{\mathrm{dom}~}}
\newcommand{\img}{{\mathrm{rng}~}}

\newcommand{\Ree}{{\mathrm{Re}~}}

\newcommand{\dis}{\displaystyle}

\begin{document}

\title{Influence of the bound states in the Neumann Laplacian in a thin waveguide}
\author{Carlos R. Mamani {\small and} Alessandra A. Verri \\ 
\vspace{-0.6cm}
\small
\em Departamento de Matem\'{a}tica -- UFSCar, \small \it S\~{a}o Carlos, SP, 13560-970 Brazil\\ \\}
\date{\today}

\maketitle

\begin{abstract} 
We study the Neumann Laplacian operator $-\Delta_{\Omega}^N$ restricted to a twisted waveguide $\Omega$.
The goal is to find the effective operator when the diameter of $\Omega$ tends to zero. 
However, when $\Omega$ is ``squeezed'' there are divergent eigenvalues
due to the transverse oscillations.
We show that each one of these eigenvalues influences the action of the effective operator in a different way. 
In the case where $\Omega$ is periodic and sufficiently thin, we find   
information about the absolutely continuous spectrum of $-\Delta_{\Omega}^N$
and the existence and location of band gaps in its structure.
\end{abstract}

\

\section{Introduction and main results}

The Laplacian operator
in a thin set with Neumann boundary
conditions has been studied in various situations \cite{prizzi,hale,michelle,ricardo}.
In particular, let $-\Delta_\Omega^N$ be the Neumann Laplacian operator restricted to a thin waveguide $\Omega$ 
in $\mathbb R^3$.
An interesting question is to study the behavior of $-\Delta_\Omega^N$ 
when the
diameter of $\Omega$ tends to zero and
to find the effective operator $T$ in this process.
Since $\Omega$ shrinks to a spatial curve, it is natural to associate $T$ with a one-dimensional operator.
In fact, it is known that $T$
is the one-dimensional Neumann Laplacian operator; in this case, its action is given by $w \mapsto -w''$. 
See, for example, \cite{ricardo}.
This result holds even if $\Omega$
is a twisted or a bent waveguide, i.e., the geometry of $\Omega$ does not influence 
the action of the effective operator.

In this work we  study $-\Delta_\Omega^N$ in the case where $\Omega$ is a twisted waveguide.
As a first goal we study its behavior as the diameter of $\Omega$ tends to zero.
In this process there are divergent eigenvalues
due to the transverse oscillations in $\Omega$.
Our strategy shows that each one of these eigenvalues influences the action of the effective operator
in a different way.
Namely, the twisted effect influences directly its action,
see (\ref{effectivepotential02}), (\ref{effectivepotential}), (\ref{effectivepotential01})
and (\ref{domintrorobin}) in this Introduction.
The second goal of this work is to consider the case where $\Omega$ is periodic in the sense that
the twisted effect varies periodically. In the case that $\Omega$ is small enough,
we find
information about the absolutely continuous spectrum of the Neumann Laplacian and the existence and location of band gaps in its structure.
In the next paragraphs, we explain the model and provide details of
our main results.

Let $I = \mathbb R$ or $I=(a,b)$ a bounded interval in $\mathbb R$.
Pick $S \neq \emptyset$ an open, bounded, smooth and connected subset of~$\mathbb R^2$;
denote by $y :=(y_1,y_2)$ an element of $S$.
Let $\alpha: \overline{I} \to \mathbb R$ be a $C^2$ function; we suppose
that $\alpha', \alpha'' \in L^\infty(I)$ and $\alpha(0) = 0$ if $ I=\mathbb R$, or $\alpha(a)=0$ if $I=(a,b)$. 
For each $\varepsilon > 0$ small enough, we define the thin twisted waveguide
\[\Omega_\varepsilon^\alpha
:=\{ \Gamma_\varepsilon^\alpha(s) {\bf x}^t, {\bf x} =(s,y) \in I \times S\},\]
where
\begin{equation}\label{gammarotation}
\Gamma_\varepsilon^\alpha(s) :=\left(\begin{array}{ccc}
1        &   0      &    0 \\
 0       &   \varepsilon   \cos \alpha(s)   &    -\varepsilon   \sin \alpha(s) \\
0         &    \varepsilon  \sin \alpha(s)   & \varepsilon   \cos \alpha(s) 
\end{array}\right).
\end{equation}

Let $-\Delta_{\Omega_\varepsilon^\alpha}^N$ be the Neumann Laplacian operator on $L^2(\Omega_\varepsilon^\alpha)$, i.e., the self-adjoint
operator  associated with the quadratic form
\begin{equation}\label{quadformneuintro}
\tilde{b}_\varepsilon (\psi) = \int_{\Omega_\varepsilon^\alpha} |\nabla \psi|^2 d \vec{x}, \quad 
\dom \tilde{b}_\varepsilon =  H^1(\Omega_\varepsilon^\alpha).
\end{equation}

Since we are going to use the $\Gamma$-convergence technique (see Appendix \ref{app1} and \cite{maso}),
our analysis is based on the study of the sequence  $(\tilde{b}_\varepsilon)_\varepsilon$.
To simplify the calculations, it is convenient
a change of variables. Using the change of coordinates described in Section \ref{geodom},
the quadratic form 
$\tilde{b}_\varepsilon$ becomes
\begin{equation}\label{enerdivostintro}
\hat{b}_\varepsilon(\psi) = \int_{Q} \left(
\left| \psi' + \langle \nabla_y \psi, R y \rangle \alpha'(s) \right|^2
+ \frac{|\nabla_y \psi|^2}{\varepsilon^2}
\right) \ds \dy,
\end{equation}
$\dom \hat{b}_\varepsilon = H^1(Q)$, $Q:=I \times S$. Here, $\psi' := \partial \psi/ \partial s$, 
$\nabla_y \psi := (\partial \psi/\partial y_1, \partial \psi/\partial y_2)$ and 
$R$ is the rotation matrix
$\left(\begin{array}{cc}
0 & -1 \\
1 & 0
\end{array}\right)$.
Denote by $-\hat{\Delta}^\varepsilon$ the self-adjoint operator associated with 
$\hat{b}^\varepsilon$.

When the waveguide is ``squeezed'', i.e, $\varepsilon \to 0$, $-\Delta^N_{\Omega_\varepsilon^\alpha}$ 
presents divergent eigenvalues due to the transverse oscillations in $\Omega_\varepsilon^\alpha$; 
one can see this by the presence of the term $(1/\varepsilon^2) \int_Q |\nabla_y \psi|^2 \ds\dy$ in 
(\ref{enerdivostintro}).
To control this divergent energies, we will take the following 
strategy:
let $-\Delta^N_S$ be the Neumann Laplacian operator restricted to $S$, i.e., the self-adjoint operator
associated with  the  quadratic form $u \mapsto \int_S |\nabla_y u|^2 \dy$, $u \in H^1(S)$.
Denote by $\lambda_n$ the $n$th eigenvalue of $-\Delta^N_S$ and by $u_n$ the corresponding normalized eigenfunction, i.e., 
\[0=\lambda_1<\lambda_2<\lambda_3<\cdots, \quad \quad 
-\Delta^N_Su_n=\lambda_nu_n,\quad n=1,2,3,\cdots.\]
We assume that each eigenvalue $\lambda_n$ is simple;
note that $u_1$ is a constant function.

Fixed $n \in \mathbb N$, our strategy 
is to study the sequence 
\[\hat{b}_n^\varepsilon(\psi) := \hat{b}_\varepsilon(\psi)-\frac{\lambda_n}{\varepsilon^2}\|\psi\|^2_{L^2(Q)},\]
$\dom \hat{b}^\varepsilon_n:= H^1(Q)$. 
Denote by $\hat{T}_n^\varepsilon$ the self-adjoint operator associated with
$\hat{b}_n^\varepsilon$; this can be done since
each quadratic form $\hat{b}_n^\varepsilon$ is closed and lower bounded in $L^2(Q)$. Namely,
$\hat{T}_n^\varepsilon = -\hat{\Delta}^\varepsilon - (\lambda_n/\varepsilon^2) \Id$;
$\Id$ denotes the identity operator.

It is usual in the literature to consider only the case $n=1$, i.e.,
since $\lambda_1=0$, to study directly the sequence of quadratic forms $\hat{b}_\varepsilon(\psi)$.
The idea to consider $n\neq 1$ 
is based on 
\cite{cesargamma}; the author considered the Dirichlet Laplacian operator restricted to
a thin waveguide with the goal of finding the effective operator. In that case, the action of the effective operator is the same 
for $n=1$ or $n\neq 1$ and depends on the geometry of the waveguide.

Now, for each $n \in \mathbb N$, consider the closed subspaces 
\[\mathcal{L}_n:=\{w(s)u_n(y):w \in L^2(I)\} \quad \hbox{and} \quad
\mathcal{K}_n:=\{w(s)u_n(y):w\in H^1(I)\}\]
of $L^2(Q)$ and $H^1(Q)$, respectively. We have the decompositions
\[L^2(Q)=\mathcal{L}_1\oplus\mathcal{L}_2\oplus\mathcal{L}_3\oplus\cdots,
\quad
H^1(Q)=\mathcal{K}_1\oplus\mathcal{K}_2\oplus\mathcal{K}_3\oplus\cdots,\]
and  each $\mathcal{K}_n$ is a dense subspace of $\mathcal{L}_n$.

Let $T_1w:= -w''$ be the one-dimensional Laplacian operator with domain
$\dom T_1 = H^2(\mathbb R)$ if $I=\mathbb R$, or $\dom T_1 = \{w \in H^2(I): w'(a)=w'(b) =0\}$ if
$I=(a,b)$. Denote by ${\bf 0}$ the null operator on the subspace
${\cal L}_1^\perp$.	
In  the particular case $n=1$, it is known that 	 
$\hat{T}^\varepsilon_1 \approx T_1 \oplus {\bf 0}$, as $\varepsilon \to 0$; see \cite{ricardo}.
As already commented, we can note that the effective operator in this situation does not depend on
the geometry of the waveguide. 

The main goal of this work is to study the sequence  
$(\hat{T}^\varepsilon_n)_\varepsilon$ (for each $n>1$ fixed), and to characterize the effective operator in the limit $\varepsilon \to 0$.
However, some adjustments will be necessary so that the limit exists in some sense.
The interesting point in this situation is that we find an effective operator that depends on
the geometry of the waveguide. To our knowledge, this fact is not known yet. 

In order to study the sequence $(\hat{T}^\varepsilon_n)_\varepsilon$,
it will be necessary some considerations.
If $v(s,y)=w(s)u_j(y)$ with $w \in H^1(I)$, some calculations show that
\begin{eqnarray*}
\hat{b}^\varepsilon_n(v) & =& \int_{Q}|w' u_j+\langle\nabla_y u_j,Ry\rangle\alpha'(s)w|^2 \ds\dy 
+ 	
\frac{1}{\varepsilon^2}\int_{Q}
\left(|\nabla_yu_j|^2-\lambda_n|u_j|^2\right)|w|^2 \ds \dy \\
& =  & \int_{Q}|w' u_j + \langle \nabla_y u_j, Ry \rangle \alpha'(s)w|^2 \ds\dy 
+ \frac{(\lambda_j-\lambda_n)}{\varepsilon^2}\|w\|^2_{L^2(I)},
\end{eqnarray*}
i.e., for $j<n$,
\begin{equation}\label{procdivexvec}
\lim\limits_{\varepsilon\rightarrow 0} \hat{b}^\varepsilon_n(v)=-\infty.
\end{equation}
Thus, the sequence $(\hat{b}^\varepsilon_n(v))_\varepsilon$ is not bounded from bellow.
Therefore, to study  $(\hat{b}_n^\varepsilon)_\varepsilon$, 
it will be necessary to exclude some vectors of the domains $\dom \hat{b}_n^\varepsilon$.
Based on (\ref{procdivexvec}), the procedure for this problem is as follows.  
We define the Hilbert spaces 
\begin{equation}\label{hilbertspaceHn}
{\cal H}_n:=\left\{\begin{array}{ll}
L^2(Q),\quad & n=1,\\
(\mathcal{L}_1\oplus\mathcal{L}_2\oplus\cdots\oplus\mathcal{L}_{n-1})^\bot, \quad & n= 2,3,\cdots,\end{array}\right.
\end{equation} 
equipped with the norm of $L^2(Q)$.
Then, we  consider the sequence of quadratic forms acting in ${\cal H}_n$;
\begin{equation}\label{quadrestintro}
\overline{b}_n^\varepsilon(\psi) = \int_Q
\left(|\psi'+\langle\nabla_y\psi,Ry\rangle\alpha'(s)|^2+\frac{1}{\varepsilon^2}|\nabla_y\psi|^2\right) \ds\dy,
\end{equation}
$\dom \overline{b}_n^\varepsilon = H^1(Q)\cap {\cal H}_n$, and we denote by $-\Delta_n^\varepsilon$ the self-adjoint operator 
on ${\cal H}_n$ associated with it.
Finally, define
\begin{equation}\label{bn}
b_n^\varepsilon(\psi) : =  \overline{b}_n^\varepsilon(\psi) - \frac{\lambda_n}{\varepsilon^2} \|\psi\|_{{\cal H}_n}^2,
\end{equation}
$\dom b^\varepsilon_n := H^1(Q)\cap {\cal H}_n$.
Denote by $T_n^\varepsilon$  the self-adjoint operator associated with $b_n^\varepsilon$ 
which is a positive and closed quadratic form; $T^\varepsilon_n$ acts in the Hilbert space ${\cal H}_n$.
Namely, 
$T_n^\varepsilon = -\Delta_n^\varepsilon - (\lambda_n/\varepsilon^2) \Id$.
Then, we are going to study the sequence $(T^\varepsilon_n)_\varepsilon$ instead of $(\hat{T}_n^\varepsilon)_\varepsilon$.

Let $b_n$ be the one-dimensional quadratic form 
\begin{equation}\label{qlastbn} 
b_n(w)  := b_n^\varepsilon(wu_n) 
=\int_{Q}|w' u_n+\langle\nabla_y u_n, Ry \rangle \alpha'(s)w|^2 \ds\dy,
\end{equation}
$\dom b_n = H^1(I)$. In fact, $b_n$ is obtained by  the restriction of
$b^\varepsilon_n$ to the space $\mathcal{K}_n$.
Denote by $T_n$ the self-adjoint operator associated with $b_n$. 

For each $n \in \mathbb N$, define the constants
\begin{equation}\label{effectivepotential02}
C_n^1(S) := \int_S|\langle\nabla_yu_n,Ry\rangle|^2 \dy, 
\quad
C_n^2(S) :=  \int_S u_n\langle\nabla_yu_n, Ry\rangle \dy,
\end{equation} 
and  the real potential 
\begin{equation}\label{effectivepotential}
V_n(s):=C^1_n(S)(\alpha'(s))^2  -  C^2_n(S)\alpha''(s).
\end{equation}
By considerations of Appendix \ref{app001},
\begin{equation}\label{effectivepotential01}
T_n w = -w'' + V_n(s) w, 
\end{equation}
where
$\dom T_n = H^2(\mathbb R)$ if $I=\mathbb R$ and,
\begin{equation}\label{domintrorobin}
\dom T_n = \left\{ w \in H^2(I): 
\begin{array}{l}
w'(a) = -C_n^2(S) \alpha'(a) w(a) \\
w'(b) = -C_n^2(S) \alpha'(b) w(b)
\end{array} \right\}
\end{equation}
if $I=(a,b)$.
In the latter, we have the Robin conditions in $\dom T_n$.

Now, we present the first result of this work.

\begin{Theorem}\label{maintheorem} 
(A) For each $n \in \mathbb N$ fixed, the sequence of self-adjoint operators $T^\varepsilon_n$ converges in the strong resolvent sense 
to $T_n$ in $\mathcal{L}_n$, as $\varepsilon\rightarrow 0$. That is,
\[\lim_{\varepsilon \to 0} R_{-\lambda} (T_n^\varepsilon) \zeta = R_{-\lambda}(T_n) P \zeta,
\quad \forall \zeta \in {\cal H}_n, \forall \lambda > 0,\]
where $P$ is the orthogonal projection onto ${\cal L}_n$.

(B) In addition, suppose $I=(a,b)$ a bounded interval. Denote by 
$\mu_j(\varepsilon)$ (resp. $\mu_j$) the $j$th eigenvalue of $-\Delta_n^\varepsilon$ (resp. $T_n$)
counted according to its multiplicity.
Then, for each $j \in \mathbb N$,
\begin{equation}\label{beasymp}
\mu_j = \lim_{\varepsilon \to 0} \left(\mu_j(\varepsilon) - \frac{\lambda_n}{\varepsilon^2}\right).
\end{equation}
\end{Theorem}

In the next paragraphs we treat the periodic case.

Consider the twisted waveguide $\Omega_\varepsilon^\alpha$
in the particular case where $I=\mathbb R$ and  $\alpha:\mathbb{R}\rightarrow\mathbb{R}$ is a $C^2$ and periodic function,  
i.e., there exists $L>0$ so that $\alpha(s+L)=\alpha(s)$, for all $s\in \mathbb R$.
The second goal of this work 
is to find spectral properties of $-\Delta_n^\varepsilon$  in this situation.
We have

\begin{Theorem}\label{theoremperiodiccase}
For each $n \in \mathbb N$ and
for each $E>0$, there exists $\varepsilon_E>0$ so that the spectrum of $-\Delta^\varepsilon_n$ is absolutely 
continuous in the interval $[0, E+\lambda_n/\varepsilon^2]$, for all $\varepsilon\in (0,\varepsilon_E)$.
\end{Theorem}

\begin{Theorem}\label{gap-exists}
Suppose that $V_n(s)$ is not a constant function in $[0,L]$.
For each $n \in \mathbb N \backslash \{1\}$,
there exist $j\in \mathbb{N}$ and $\varepsilon_j > 0$,  so that, for all $\varepsilon\in (0,\varepsilon_j)$, the spectrum of the operator $-\Delta^\varepsilon_n$ has at least one gap. 
\end{Theorem}

Furthermore, in Theorem \ref{location} in Section \ref{locationbandsec}, we find a location in $\sigma(-\Delta_n^\varepsilon)$
where Theorem \ref{gap-exists} holds true. 

The proof of Theorem \ref{gap-exists} is based on the fact that $V_n(s)$ is not a constant function in $[0,L]$. 
Due to this reason, we eliminate the case $n=1$ since $V_1(s) \equiv 0$.
	
This work is separated as follows.
In Section \ref{geodom} we perform the change of variables to obtain (\ref{enerdivostintro})
and in Section \ref{proof} we prove Theorem \ref{maintheorem}.
Section \ref{specprsecsub} is dedicated to the periodic case and is separated in subsections.
In Subsection \ref{preressub} we present some preliminary results and in Subsection 
\ref{prooftheosub} we prove Theorem \ref{theoremperiodiccase}.
Subsection \ref{subsection-gaps} is dedicated to prove Theorem \ref{gap-exists} and in
Subsection \ref{locationbandsec} we study the location of band gaps.
A long the text, the symbol $K$ is used to denote different constants and it never depends
on $\theta$.

\section{Geometry of the domain}\label{geodom}

Recall the quadratic form $\tilde{b}_\varepsilon$ defined by (\ref{quadformneuintro}).
In this section we  perform a usual change of variables so that  the domain $\dom \tilde{b}_\varepsilon$  
becomes independent of $\varepsilon$. 
Then, consider the mapping
$$\begin{array}{cccl}
F_{\varepsilon}: & Q & \rightarrow & \Omega_{\varepsilon}^\alpha\\
& (s,y_1,y_2) & \mapsto & \Gamma^\alpha_\varepsilon(s)(s,y_1,y_2)^t
\end{array},$$
where $\Gamma^\alpha_\varepsilon(s)$ is given by $(\ref{gammarotation})$;
$F_{\varepsilon}$ will be a (global) diffeomorphism for  $\varepsilon  > 0$ small enough.
 
In the new variables the domain $\dom \tilde{b}_\varepsilon$ turns to be $H^1(Q)$. On the other hand, the price to be paid is a nontrivial Riemannian metric $G=G^{\alpha}_{\varepsilon}$ which is induced by $F_\varepsilon$ i.e., 
$G=(G_{ij})$, $G_{ij}=\langle e_i,e_j\rangle$, $1\leq i, j\leq 3$,
where $e_1=\partial F_{\varepsilon}/\partial s,$ $e_2=\partial F_{\varepsilon}/\partial y_1$, 
and $e_3= \partial F_{\varepsilon}/\partial y_2$.
Some calculations show that 
\[J=\left(\begin{matrix}
e_1\\
e_2\\
e_3
\end{matrix}\right)=\left(\begin{matrix}
1 & -\varepsilon \alpha'(s)\langle z^{\bot}_{\alpha}(s),y\rangle & \varepsilon \alpha'(s)\langle z_{\alpha}(s),y\rangle,\\
0 & \varepsilon \cos\alpha(s) & \varepsilon \sin \alpha(s) \\
0 & -\varepsilon \sin \alpha(s) & \varepsilon \cos\alpha(s)
\end{matrix}\right),\]
where  
\[z_{\alpha}(s) := (\cos\alpha(s),-\sin\alpha(s)), \quad z^{\bot}_{\alpha}(s):=(\sin\alpha(s),\cos\alpha(s)).\]
The inverse matrix of $J$ is given by 
$$J^{-1}=\left( \begin{matrix}
1 & \alpha'(s)y_2 & -\alpha'(s)y_1\\
0 & (\cos\alpha(s))/\varepsilon & -(\sin\alpha(s))/\varepsilon\\
0 & (\sin\alpha(s))/\varepsilon & (\cos\alpha(s))/\varepsilon 
\end{matrix}\right).$$

Note that $JJ^{t}=G$ and $\det J=|\det G|^{1/2}=\varepsilon^2>0$. Thus, $F_\varepsilon$ is a local diffeomorphism. 
In the case that $F_\varepsilon$ is injective (for this, just consider $\varepsilon > 0$ small enough), 
a global diffeomorphism  is obtained. 

Introducing the unitary transformation
\[\begin{array}{cccc}
U_\varepsilon: & {L}^2(\Omega^\alpha_\varepsilon) & \rightarrow & {L}^2(Q)\\
& \phi & \mapsto  & \varepsilon \phi \circ F_\varepsilon
\end{array},\]
we obtain the quadratic form 
\begin{eqnarray*}
\hat{b}_\varepsilon(\psi) &:=& \tilde{b}_\varepsilon(U_\varepsilon^{-1}\psi) = \|J^{-1}\nabla\psi\|_{L^2(Q)}^2 \\
&=&\int_{Q} \left(
\left| \psi' + \langle \nabla_y \psi, R y \rangle \alpha'(s) \right|^2
+ \frac{|\nabla_y \psi|^2}{\varepsilon^2}
\right) \ds \dy,
\end{eqnarray*}
$\dom \hat{b}_\varepsilon = H^1(Q)$.
Recall $R$  is the rotation matrix $\left(\begin{array}{cc}
0 & 1\\
-1 & 0 
\end{array}\right)$  and $-\hat{\Delta}_\varepsilon$ denotes the self-adjoint operator associated with $\hat{b}_\varepsilon$.
We have $ -\hat{\Delta}_\varepsilon \psi =
U_\varepsilon(-\Delta^N_{\Omega_\varepsilon^\alpha})U^{-1}_\varepsilon \psi$, where 
$\dom (-\hat{\Delta}_\varepsilon) = U_\varepsilon( \dom (-\Delta^N_{\Omega_\varepsilon^\alpha}) )$.

\section{Preliminary results and poof of Theorem \ref{maintheorem}}\label{proof}

This section is dedicated to the proof of Theorem \ref{maintheorem}. The strategy is based on the
study of the sequence $(b_n^\varepsilon)_\varepsilon$ (see (\ref{bn}) in the Introduction) and some preliminary results will be necessary. 
We start with some considerations.
Denote  by $[u_1,u_2,\dots,u_k]$ the subspace of $L^2(S)$ generated by $\{u_1,u_2,\dots, u_k\}$. 
Since the subspace ${\cal W}_k := [u_1,u_2,\dots,u_k]^\bot$  is invariant under the operator $-\Delta^N_S$, 
the restriction $-\Delta^N_S|_{{\cal W}_k}$ is well defined and its first eigenvalue  is $\lambda_{k+1}$. 
Denote by $q_k$  the quadratic form associated with $-\Delta^N_S|_{{\cal W}_k}$. 
We have
\begin{equation}\label{minmaxins}
q_k(v)\geq \lambda_{k+1}\|v\|^2_{L^2(S)},\quad \forall v \in {\cal W}_k\cap H^1(S).
\end{equation}

To study the sequence $(b_n^\varepsilon)_\varepsilon$ we are going to
use the $\Gamma$-convergence technique; see Appendix \ref{app1}.
Then, it is necessary to extend  $b^\varepsilon_n$  on ${\cal H}_n$ by setting (we denote by the same symbol)
$$b_n^\varepsilon(v)=\left\{\begin{array}{lc}
b_n^\varepsilon(v),& \text{if} \,v\in {\rm dom }\, b^\varepsilon_n,\\
+\infty,& \text{otherwise}.
\end{array}\right.$$	
In a similar way, we extend $b_n$ on ${\cal H}_n$;
$$b_n(v)=\left\{\begin{array}{cc}
b_n(w), \quad &\text{if}\,\, v=wu_n \,\text{ with }\, w\in {\rm dom}\, b_n,\\
+\infty, \quad &\text{otherwise};
\end{array}\right.$$
recall the definition of $b_n$ by (\ref{qlastbn}) in the Introduction.

\begin{Lemma}\label{lem1}
If $v_\varepsilon \rightharpoonup v$ in ${\cal H}_n$ and $(b_n^\varepsilon(v_\varepsilon))_\varepsilon$ is a bounded sequence, 
then $(v'_\varepsilon)_\varepsilon$ and $(\nabla_yv_\varepsilon)_\varepsilon$ are bounded sequences in ${\cal H}_n$. 
Furthermore, $v\in H^1(Q)$ and there exists a subsequence of  $(v_\varepsilon)_\varepsilon$, denoted by the same symbol
$(v_\varepsilon)_\varepsilon$, so that $v'_\varepsilon\rightharpoonup  v'$ and $\nabla_yv_\varepsilon\rightharpoonup\nabla_yv$.
\end{Lemma}
\begin{proof}
Since $(v_\varepsilon)_\varepsilon$ and $(b_n^\varepsilon(v_\varepsilon))_\varepsilon$  are bounded sequences, 
there exists a number $K>0$ so that 
\[\limsup\limits_{\varepsilon\rightarrow0}\int_{Q}|v'_\varepsilon+\langle\nabla_yv_\varepsilon,Ry\rangle\alpha'(s)|^2 \ds \dy
\leq \limsup\limits_{\varepsilon\rightarrow0}b_n^\varepsilon(v_\varepsilon) <  K,\]
and
\begin{eqnarray}\label{intcdl}
& & \limsup\limits_{\varepsilon\rightarrow 0}\int_{Q}|\nabla_yv_\varepsilon|^2 \ds\dy \nonumber\\
& = & \limsup\limits_{\varepsilon\rightarrow 0}\left(\int_{Q} \left(|\nabla_y v_\varepsilon|^2-\lambda_n|v_\varepsilon|^2\right) \ds\dy 
+\int_{Q} \lambda_n |v_\varepsilon|^2 \ds\dy\right) \nonumber\\
& \leq & \limsup\limits_{\varepsilon\rightarrow 0} K \varepsilon^2+
\limsup\limits_{\varepsilon\rightarrow 0}\int_{Q} \lambda_n|v_\varepsilon|^2 \ds \dy< K.
\end{eqnarray}

These estimates, and the fact that $\alpha'$ and $R\,y$ are bounded functions, show that
$(v'_\varepsilon)_\varepsilon$ and $(\nabla_y v_\varepsilon)_\varepsilon$ are bounded sequences in $L^2(Q)$.
Therefore, $(v_\varepsilon)_\varepsilon$ is a bounded sequence in $H^1(Q)$. 
Thus, there exists $\psi\in H^1(Q)$ and a subsequence of $(v_\varepsilon)_\varepsilon$,  
also denoted by $(v_\varepsilon)_\varepsilon$, so that $v_\varepsilon\rightharpoonup \psi$ in $H^1(Q)$ 
(recall that this Hilbert space is reflexive). As $v_\varepsilon\rightharpoonup v$ in ${\cal H}_n$, it follows that $v=\psi$,   
$v'_\varepsilon\rightharpoonup v'$, $\nabla_yv_\varepsilon\rightharpoonup \nabla_y v$ in ${\cal H}_n$ and $v\in H^1(Q)$.
\end{proof}

\begin{Lemma}\label{lem2}
If $v_\varepsilon\rightarrow v$ in ${\cal H}_n$ and
there exists the limit $\lim\limits_{\varepsilon\rightarrow 0} b_n^\varepsilon(v_\varepsilon) < +\infty$, then
$v(s,y)=w(s)u_n(y)$ with $w\in H^1(I)$ (i.e., $v \in {\cal K}_n$).
\end{Lemma}
\begin{proof}
By Lemma \ref{lem1}, passing to a subsequence if necessary,  $\nabla_yv_\varepsilon\rightharpoonup \nabla_y v$  in $L^2(Q)$.  
By weak lower semi-continuity of the $L^2$-norm, 
inequality (\ref{intcdl}) and from the strong convergence of 
$(v_\varepsilon)_\varepsilon $, we have
\[\int_{Q}|\nabla_y v|^2 \ds\dy \leq \liminf\limits_{\varepsilon\rightarrow 0}\int_{Q}|
\nabla_yv_\varepsilon|^2 \ds \dy
\leq \limsup\limits_{\varepsilon\rightarrow 0}\lambda_n\int_{Q}|v_\varepsilon|^2 \ds\dy
 = \lambda_n\int_{Q}|v|^2 \ds\dy.\]

Now, define the function
$f_n(s):= \int_S\left(|\nabla_yv(s,y)|^2-\lambda_n|v(s,y)|^2\right) \dy.$ 
The latter inequalities show that $f_n(s) \leq 0$. However,
(\ref{minmaxins}) ensures that $f_n(s)\geq 0$. 
Then,  $f_n=0$ a.e..
We conclude that $v(s,\cdot)\in {\cal W}_{n-1}\cap H^1(S)$, and  $v(s,\cdot)$ is an eigenfunction of the operator  
$-\Delta^N_S|_{{\cal W}_{n-1}}$ whose eigenvalue associated is $\lambda_n$. 
As $\lambda_n$ is simple, $v(s,\cdot)$ is proportional to $u_n$. 
Thus, we can write $v(s,y)=w(s)u_n(y)$ with $w\in H^1(I)$, since $v\in H^1(Q)$. 
\end{proof}

\begin{Proposition}\label{propstrong}
For each $n \in \mathbb N$, the sequence of quadratic forms $b_n^\varepsilon$ strongly 
$\Gamma$-converges to $b_n$, as $\varepsilon \to 0$.
\end{Proposition}	

\begin{proof}
We have to prove the itens $(i)$ and $(ii)$ 
according to the definition of strong $\Gamma$-convergence in Appendix \ref{app1}.
 
Let $v\in {\cal H}_n$ and $v_\varepsilon\rightarrow v$ in ${\cal H}_n$. 
If $\liminf\limits_{\varepsilon\rightarrow0}b_n^\varepsilon(v_\varepsilon)=+\infty$, 
then $b_n(v)\leq\liminf\limits_{\varepsilon\rightarrow0}b_n^\varepsilon(v_\varepsilon).$ 
Now, assume that $\liminf\limits_{\varepsilon\rightarrow0}b_n^\varepsilon(v_\varepsilon)<+\infty$. 
Passing to a subsequence if necessary, 
we can suppose that $\liminf\limits_{\varepsilon\rightarrow0}b_n^\varepsilon(v_\varepsilon)=
\lim\limits_{\varepsilon\rightarrow0}b_n^\varepsilon(v_\varepsilon)<+\infty.$

Lemma \ref{lem1} ensures that $v'_\varepsilon\rightharpoonup v'$, 
$\nabla_yv_\varepsilon\rightharpoonup \nabla_yv$ in $L^2(Q)$, and $v \in H^1(Q)$. 
Since $\alpha'$ is a bounded function, 
\[v'_\varepsilon+\langle\nabla_yv_\varepsilon,Ry\rangle\alpha'\rightharpoonup v'+\langle\nabla_yv,Ry\rangle\alpha'\]
in $L^2(Q)$. 
Then, 
\begin{eqnarray*}
\int_{Q}|v'+\langle\nabla_yv,Ry\rangle\alpha'(s)|^2 \ds\dy & \leq& \liminf\limits_{\varepsilon\rightarrow0}\int_{Q}|v'_\varepsilon+\langle\nabla_yv_\varepsilon,Ry\rangle\alpha'(s)|^2 \ds\dy\\
&\leq& \liminf\limits_{\varepsilon\rightarrow0}b_n^\varepsilon(v_\varepsilon).
\end{eqnarray*}
	
By Lemma \ref{lem2}, we can write $v(s,y)=w(s)u_n(y)$ with $w\in H^1(I)$. Thus,	
\[b_n(w) = b_n(v)\leq \liminf\limits_{\varepsilon\rightarrow0}b_n^\varepsilon(v_\varepsilon),\] 
and item $(i)$ is proven.

To prove $(ii)$, we are going to show that for each $v \in {\cal H}_n$ there exists a sequence 
$(v_\varepsilon)_\varepsilon$ in ${\cal H}_n$ so 
that $v_\varepsilon\rightarrow v$  in ${\cal H}_n$ and $\lim\limits_{\varepsilon\rightarrow 0}b^\varepsilon_n(v_\varepsilon)=b_n(v)$. 
At first, consider the particular  case  $v(s,y)=w(s)u_n(y)$ with $w\in H^1(I)$. 
Take $v_\varepsilon:=v$, for all $\varepsilon> 0$. Note that $b_n^\varepsilon(v)=b_n(w)$, for all $\varepsilon > 0$, and 
\[\lim\limits_{\varepsilon\rightarrow0}b_n^\varepsilon(v_\varepsilon)=b_n(v).\]
	
On the other hand, if $v\in {\cal H}_n\setminus\{w(s)u_n(y):w\in H^1(I)\}$, one has
$b_n(v)=+\infty$. Let $(v_\varepsilon)_\varepsilon$ be an arbitrary sequence so that $v_\varepsilon\rightarrow v$ in ${\cal H}_n$. 
In this case, $\dis\liminf_{\varepsilon\rightarrow0}b_n^\varepsilon(v_\varepsilon)=+\infty$. 
In fact, if we suppose that $\liminf\limits_{\varepsilon\rightarrow0}b_n^\varepsilon(v_\varepsilon)<+\infty$, 
by Lemmas \ref{lem1} and \ref{lem2} we should have $v=wu_n$, with $w\in H^1(I)$, 
but this is not true. 
Therefore, $+\infty = \liminf\limits_{\varepsilon\rightarrow0}b_n^\varepsilon(v_\varepsilon)=
\lim\limits_{\varepsilon\rightarrow0}b_n^\varepsilon(v_\varepsilon) =b_n(v)$.
Then, item $(ii)$ is satisfied.
\end{proof}

\begin{Proposition}\label{propweaky}
For each $n \in \mathbb N$, the sequence of quadratic forms $b_n^\varepsilon$ weakly 
$\Gamma$-converges to $b_n$, as $\varepsilon \to 0$.
\end{Proposition}
\begin{proof}
At first, we are going to show the condition $(i)$ of the definition of weak $\Gamma$-convergence, i.e.,
$b_n(v)\leq \liminf\limits_{\varepsilon\rightarrow0}b_n(v_\varepsilon)$, for all sequence 
$v_\varepsilon\rightharpoonup v$ in ${\cal H}_n$.
So, assume the weak convergence  $v_\varepsilon\rightharpoonup v$.
Initially, consider the case where $(v_\varepsilon)_\varepsilon$ does not belong to ${\cal H}_n \cap H^1(Q)$.
Then,  $b_n^\varepsilon(v_\varepsilon) = + \infty$, for all $\varepsilon > 0$,
and the inequality is proven.
Now, assume that $(v_\varepsilon)_\varepsilon \subset {\cal H}_n \cap H^1(Q)$.
Suppose  
that  $v=wu_n$, with $w \in H^1(I)$.
By definition, $b_n(v)<+\infty$.
If $\liminf\limits_{\varepsilon\rightarrow 0}b^\varepsilon_n(v_\varepsilon)=+\infty$ the inequality is proven.
Now, suppose that $\liminf\limits_{\varepsilon\rightarrow 0}b_n^\varepsilon(v_\varepsilon)<+\infty$. Passing to a subsequence
if necessary, we can suppose that
$\liminf\limits_{\varepsilon\rightarrow0}b_n^\varepsilon(v_\varepsilon)=\lim\limits_{\varepsilon\rightarrow0}b_n^\varepsilon(v_\varepsilon)<+\infty.$	
As in the proof of Proposition \ref{propstrong},
\begin{eqnarray*}
\lim\limits_{\varepsilon\rightarrow0}b_n^\varepsilon(v_\varepsilon)  
& \geq & \int_{Q}|v'+\langle\nabla_yv,Ry\rangle\alpha'(s)|^2 \ds \dy \\
& = & b_n(w).
\end{eqnarray*}

Now, suppose that $v$ does not belong
to the subspace $\{w u_n: w\in H^1(I) \}$.
We are going to show that necessarily 
$\displaystyle\liminf_{\varepsilon\rightarrow0}b_n^\varepsilon(v_\varepsilon)=+\infty$. 
In fact, let $P_{n+1}$ be the orthogonal projection onto ${\cal H}_{n+1}$.
We have $\|P_{n+1}v\|>0$.
Since $v_\varepsilon \rightharpoonup v$  in ${\cal H}_n \cap H^1(Q)$,
$P_{n+1}v_\varepsilon\rightharpoonup P_{n+1}v$ and
\begin{equation}\label{projinorth}
\liminf\limits_{\varepsilon\rightarrow0}\|P_{n+1}v_\varepsilon\|\geq \|P_{n+1}v\|>0.
\end{equation}
	
Note that
\begin{equation}\label{infheldiv}
b_n^\varepsilon(v_\varepsilon)\geq\frac{1}{\varepsilon^2}\int_{Q}
\left(|\nabla_yv_\varepsilon|^2-\lambda_n|v_\varepsilon|^2\right)\ds\dy.
\end{equation}
The strategy is to estimate the term on the right side of this inequality.
	
For $\psi\in H^1(S)\cap {\cal W}_{n-1}$,
denote by $\psi^n$ the component of $\psi$ in $[u_n]$ and by $Q_{n+1}$ the orthogonal projection onto 
${\cal W}_n$ in $H^1(S)$. Thus,
\begin{eqnarray*}
\frac{1}{\varepsilon^2}\int_{Q} \left(|\nabla_yv_\varepsilon|^2 -\lambda_n|v_\varepsilon|^2\right) \ds\dy
& = &
\frac{1}{\varepsilon^2}\int_{I}\left(\|\nabla_yv_\varepsilon(s,\cdot)\|^2_{L^2(S)}-\lambda_n\|v_\varepsilon(s,\cdot)\|^2_{L^2(S)}\right) \ds\\
& = &  
\frac{1}{\varepsilon^2}\int_{I}\left(\|v_\varepsilon(s,\cdot)\|^2_{H^1(S)}-(\lambda_n+1)\|v_\varepsilon(s,\cdot)\|^2_{L^2(S)}\right) \ds\\
& = &  
\frac{1}{\varepsilon^2}\int_{I}\left(\|Q_{n+1}v_\varepsilon(s,\cdot)\|^2_{H^1(S)} +\|v^n_\varepsilon(s,\cdot)\|^2_{H^1(S)}\right. \\
& - &
\left.(\lambda_n+1)\|v_\varepsilon(s,\cdot)\|^2_{L^2(S)}\right)\ds\\
& = &
\frac{1}{\varepsilon^2}\int_{I}\left(\|\nabla_yQ_{n+1}v_\varepsilon(s,\cdot)\|^2_{L^2(S)} +
\|Q_{n+1}v_\varepsilon(s,\cdot)\|^2_{L^2(S)}\right.\\
& + & 
\|\nabla_nv^n_\varepsilon(s,\cdot)\|^2_{L^2(S)}+\|v^n_\varepsilon(s,\cdot)\|^2_{L^2(S)}\\
& - &
\left.(\lambda_n+1)\|v_\varepsilon(s,\cdot)\|^2_{L^2(S)}\right) \ds\\
&\geq&  
\frac{1}{\varepsilon^2}\int_{I}\left(\lambda_{n+1}\|Q_{n+1}v_\varepsilon(s,\cdot)\|^2_{L^2(S)} +
\lambda_n\|v^n_\varepsilon(s,\cdot)\|^2_{H^1(S)}\right.\\
& - &
\left.\lambda_n\|v_\varepsilon(s,\cdot)\|^2_{L^2(S)}\right)\ds\\
& = &
\frac{1}{\varepsilon^2}\int_{I} (\lambda_{n+1}-\lambda_n)|Q_{n+1} v_\varepsilon|^2\ds\dy\\
& = & 
\frac{(\lambda_{n+1}-\lambda_n)}{\varepsilon^2}\|P_{n+1}v_\varepsilon\|^2\\
& \geq &
\frac{(\lambda_{n+1}-\lambda_n)}{\varepsilon^2}\|P_{n+1}v\|^2.
\end{eqnarray*}
This estimate, (\ref{projinorth}), (\ref{infheldiv}) and the fact that
$\lambda_{n+1} > \lambda_n$, imply that
$\liminf\limits_{\varepsilon\rightarrow0}b_n^\varepsilon(v_\varepsilon) = + \infty$.

Finally, the condition $(ii)$ of the definition of weak 
$\Gamma$-convergence can be proven in a similar way to the proof of 
Proposition \ref{propstrong}.
\end{proof}

\vspace{0.3cm}
\noindent
{\bf Proof of Theorema \ref{maintheorem}}:
$(A)$ This item follows by
Propositions \ref{propstrong} and \ref{propweaky} of this section and 
Proposition \ref{propappfirst} in Appendix \ref{app1}.

\vspace{0.3cm}
\noindent
$(B)$	We have to verify the itens $a), b),$ and $c)$ of Propostion \ref{propappsecond} in Appendix \ref{app1}.
Item $a)$ follows by Propositions \ref{propstrong} and \ref{propweaky}. It is known that
the operator $T_n$ has compact resolvent.
Thus, $b)$ is satisfied.
It remains to ensure $c)$.
Consider the subspace $\mathcal{K} := \{\mathcal{K}_1\oplus \dots \oplus\mathcal{K}_{n-1}\}^\perp$.
By  Rellich-Kondrachov Theorem, $\mathcal{K}$ is compactly embedded in ${\cal H}_n$. 
Thus, if $(v_\varepsilon)_\varepsilon$ is a bounded sequence in ${\cal H}_n$ and $(b^\varepsilon_n(v_\varepsilon))_\varepsilon$ is also bounded, a similar proof to the Lemma \ref{lem1}  shows that $(v_\varepsilon)_\varepsilon$ is a bounded sequence in $\mathcal{K}$.
So, item $c)$ is satisfied. 
By Proposition \ref{propappsecond} in Appendix \ref{app1},  
$T^\varepsilon_n$ converges in the norm resolvent sense to $T_n$ in ${\cal L}_n$.
By Corollary 2.3 in \cite{gohberg}, we have the asymptotic behavior of the eigenvalues given by (\ref{beasymp}).


\section{Spectral properties in the case of periodic waveguide}\label{specprsecsub}

Consider $\Omega_\varepsilon^\alpha$ as in the Introduction in the particular case where $I=\mathbb R$
and
$\alpha: \mathbb R \to \mathbb R$ is a $C^2$ and periodic function, i.e., there exists $L>0$ so that
$\alpha(s+L) = \alpha(s)$, for all $s \in \mathbb R$.
In this context, the goal of this section is to find spectral information about the spectrum
of $-\Delta_n^\varepsilon$, for each $n \in \mathbb N$.
Namely, we study the continuous absolutely spectrum  $\sigma_{ac}(-\Delta_n^\varepsilon)$
and
the existence and location of band gaps in $\sigma(-\Delta_n^\varepsilon)$.

\subsection{Preliminary results}\label{preressub}
		
Due to the periodic characteristics of $-\Delta_n^\varepsilon$,
to prove Theorems \ref{theoremperiodiccase} and \ref{gap-exists}, 
we are going to use the Floquet-Bloch reduction under the Brillouin zone $\mathcal{C}=[-\pi/L,\pi/L)$.
More precisely, define $Q_L:=(0,L)\times S$, ${\cal L}_n^L:=\{w(s)u_n(y): w \in L^2(0,L)\}$, $n \in \mathbb N$,
\[{\cal H}^L_n := \left\{
\begin{array}{ll}
L^2(Q_L), & n=1,\\
({\cal L}^L_1 \oplus {\cal L}^L_2 \oplus \cdots \oplus {\cal L}^L_{n-1})^\perp, & n=2,3,\cdots.
\end{array}\right.\] 
Consider the family of quadratic forms acting in ${\cal H}^L_n$:
\begin{equation}\label{fibersquaform}
\hat{b}^\varepsilon_n(\theta)(\varphi) = \int_{Q_L} \left(
\left| \varphi' +i\theta \varphi+ \langle \nabla_y \varphi , R y \rangle \alpha'(s) \right|^2
+ \frac{|\nabla_y \varphi |^2}{\varepsilon^2}
\right) \ds \dy, \quad \theta\in\mathcal{C},
\end{equation}
$\dom \hat{b}_n^\varepsilon(\theta) =\{\varphi \in H^1(Q_L)\cap {\cal H}^L_n;
\varphi (0,\cdot)=\varphi (L,\cdot) \text{\,in\,}L^2(S)\}$. 	
Denote by $-\Delta_n^\varepsilon(\theta)$ the self-adjoint operator associated with $\hat{b}^\varepsilon_n(\theta)$.

\begin{Lemma}\label{analytic}
For each $n \in \mathbb N$,
$\{-\Delta_n^\varepsilon(\theta), \theta\in \mathcal{C}\} $ is an analytic family of type $(B)$. 
\end{Lemma}
\begin{proof} 
At first, note that
$\dom \hat{b}^\varepsilon_n(\theta)$ 
does not depend on $\theta$. 
For each $\theta \in \mathcal{C}$, write $\hat{b}^\varepsilon_n(\theta)= \hat{b}^\varepsilon_n(0)+ c^\varepsilon_n(\theta)$,  where, 
for $\varphi \in \dom \hat{b}^\varepsilon_n(0)$,
\begin{eqnarray*}
c^\varepsilon_n(\theta)(\varphi)&:=& \hat{b}^\varepsilon_n(\theta)(\varphi)-\hat{b}^\varepsilon_n(0)(\varphi)\\
&=& 2 \, \Ree \left(\int_{Q_L}\overline{\left(\varphi'+\langle\nabla_y \varphi, R_y\rangle\alpha'(s)\right)}
(i\theta \varphi) \ds\dy \right) + \theta^2\int_{Q_L}|\varphi|^2 \ds\dy.
\end{eqnarray*} 

We affirm that $c^\varepsilon_n(\theta)$ is $\hat{b}^\varepsilon_n(0)$-bounded with zero relative bound. In fact, 
given $\delta> 0$, 
\begin{eqnarray*}
|c^\varepsilon_n(\theta)(\varphi)|   
& \leq &  2 \int_{Q_L}|\varphi'+\langle\nabla_y \varphi,Ry\rangle \alpha'(s)|\,
|i\theta \varphi| \ds\dy + \theta^2\int_{Q_L}| \varphi |^2 \ds \dy\\
& \leq &  \delta \int_{Q_L}| \varphi'+\langle\nabla_y \varphi,Ry\rangle\alpha'(s)|^2 \ds\dy + 
\theta^2(1/ \delta + 1)\int_{Q_L}| \varphi |^2 \ds\dy\\
&\leq & \delta  \, \hat{b}^\varepsilon_n(0)(\varphi) + (\pi / L)^2(1/\delta+1)\|\varphi\|_{{\cal H}_n^L}^2,
\end{eqnarray*}
for all $\varphi \in \dom \hat{b}_n^\varepsilon(0)$, for all $\theta\in \mathcal{C}$. 
Since $\delta > 0$ is arbitrary, the affirmation is proven. 
By  Theorem 4.8, Chapter VII
in \cite{kato}, $\{\hat{b}^\varepsilon_n(\theta):\theta\in \mathcal{C}\}$ is an analytic family of type $(a)$. 
Consequently, $\{-\Delta_n^\varepsilon (\theta), \theta\in \mathcal{C}\} $ is an analytic family of type $(B)$. 	 
 \end{proof}

\begin{Lemma}\label{unitary-operator}
There exists a unitary operator 
$\mathcal{U}_n: {\cal H}_n\rightarrow\int_{\mathcal{C}}^\oplus {\cal H}^L_n\, {\rm d} \theta$, so that,
\[\mathcal{U}_n (-\Delta_n^\varepsilon) \mathcal{U}^{-1}_n=\int_{\mathcal{C}}^{\oplus}
-\Delta_n^\varepsilon(\theta)\, {\rm d}\theta.\]
\end{Lemma}
\begin{proof}
For $(\theta,s,y)\in \mathcal{C}\times Q_L$, define
\[(\mathcal{U}_n f)(\theta,s,y):= \sum_{k\in \mathbb{Z}}\sqrt{\frac{L}{2\pi}}e^{-ikL\theta-i\theta s}f(s+Lk,y),
\quad 
\dom {\cal U}_n = {\cal H}_n,\]
which is a unitary operator onto 
$\int_{\mathcal{C}}^\oplus {\cal H}^L_n\, {\rm d}\theta$; the definition of ${\cal U}_n$ is based on \cite{bde, reedsimon4} .
	 
Recall the quadratic form $\overline{b}^\varepsilon_n$; see (\ref{quadrestintro}) in the Introduction. Consider 
\[q^\varepsilon_n(\varphi):= \overline{b}^\varepsilon_n(\mathcal{U}_n^{-1}\varphi), \quad \quad
\dom q^\varepsilon_n:= {\cal U}_n(\dom \overline{b}^\varepsilon_n).\]
Note that  $q^\varepsilon_n$ is a closed and bounded from below quadratic form in the Hilbert space 
$\int_{\mathcal{C}}^\oplus {\cal H}^L_n\, {\rm d}\theta$, and
$\mathcal{U}_n (-\Delta_n^\varepsilon) \mathcal{U}^{-1}_n$ is 
the self-adjoint operator associated with it. 
	
For $(s,y) \in Q_L$ and $k\in \mathbb{Z}$,
\[ (\mathcal{U}_n^{-1}\varphi)(s+Lk, y)=\int_{\mathcal{C}}  \sqrt{\frac{L}{2\pi}} 
e^{ikL\theta +is\theta}\varphi(\theta,s,y) \, {\rm d}\theta,\]
\[ (\mathcal{U}_n^{-1}\varphi)'(s+ Lk, y)=\int_{\mathcal{C}}  \sqrt{\frac{L}{2\pi}} 
e^{ikL\theta +is\theta}(\varphi'(\theta,s,y)+i\theta\varphi(\theta,s,y))\,{\rm d}\theta,\]
and
\[ \nabla_y(\mathcal{U}_n^{-1}\varphi)(s+ L k, y)=\int_{\mathcal{C}}  \sqrt{\frac{L}{2\pi}} 
e^{ikL\theta +is\theta}\nabla_y\varphi(\theta,s,y)\,{\rm d}\theta.\]
Since $\alpha'$ is an $L$-periodic function,  by Parseval's identity, and by Fubini's Theorem, we have 
\begin{eqnarray*}
q^\varepsilon_n(\varphi)&=& \overline{b}^\varepsilon_n(\mathcal{U}_n^{-1}\varphi)\\
& = &\int_{Q} \left(
	\left| (\mathcal{U}_n^{-1}\varphi)' + \langle \nabla_y (\mathcal{U}_n^{-1}\varphi), R y \rangle \alpha'(s) \right|^2
	+ \frac{|\nabla_y (\mathcal{U}_n^{-1}\varphi)|^2}{\varepsilon^2}
	\right) \ds \dy \\
	&=&\sum_{k \in \mathbb{Z}}\int_{Q_L} 
	\left| (\mathcal{U}_n^{-1}\varphi)'(s+Lk,y) + 
	\langle \nabla_y (\mathcal{U}_n^{-1}\varphi)(s+Lk,y), R y \rangle \alpha'(s) \right|^2\,\ds \dy \\
	&+& \sum_{k\in \mathbb{Z}}\int_{Q_L}\frac{1}{\varepsilon^2}\left|\nabla_y (\mathcal{U}_n^{-1}\varphi)(s+Lk,y)\right|^2
	\ds \dy \\
	&=&\int_{Q_L} \sum_{k\in \mathbb{Z}}
	\left| \int_\mathcal{C}\sqrt{\frac{L}{2\pi}}e^{ikL\theta+is\theta}(\varphi'(\theta,s,y) +i\theta
\varphi(\theta,s,y) + \langle \nabla_y \varphi(\theta,s,y), R y \rangle \alpha'(s))\mathrm{d}\theta \right|^2\ds \dy \\
	&+& \int_{Q_L}\sum_{k\in \mathbb{Z}}\frac{1}{\varepsilon^2}\left|\int_\mathcal{C}\sqrt{\frac{L}{2\pi}}e^{ikL\theta+is\theta}\nabla_y 
	\varphi(\theta,s,y) {\rm d} \theta \right|^2
	\ds \dy \\	
	&=&\int_{Q_L} \left( \int_\mathcal{C}\left| (\varphi'(\theta,s,y) +i\theta
\varphi (\theta,s,y) + \langle \nabla_y \varphi(\theta,s,y), R y \rangle \alpha'(s)) \right|^2 {\rm d}\theta \right) \ds \dy \\
	&+& \int_{Q_L}\left( \int_\mathcal{C}\frac{1}{\varepsilon^2}\left|\nabla_y \varphi(\theta,s,y)d\theta\right|^2 {\rm d} \, \theta \right) \ds \dy\\
	& = &
	\int_{\cal C} \left( \int_{Q_L}\left| (\varphi'(\theta,s,y) +i\theta
\varphi (\theta,s,y) + \langle \nabla_y \varphi(\theta,s,y), R y \rangle \alpha'(s)) \right|^2 \ds \dy \right) {\rm d}\theta  \\
	&+& \int_{\cal C} \left( \int_{Q_L}\frac{1}{\varepsilon^2}\left|\nabla_y \varphi(\theta,s,y)d\theta\right|^2 \ds \dy \right) {\rm d}\theta \\
	&=:&\int_{\mathcal{C}}\hat{b}^\varepsilon_n(\theta)(\varphi(\theta))\, {\rm d} \theta.
\end{eqnarray*}		
Then, 
$\varphi \in \dom q^\varepsilon_n$ if, and only if, $\varphi \in \int_{{\cal C}}^\oplus {\cal H}_n^L {\rm d} \theta$ and
$\varphi(\theta)\in \dom \hat{b}^\varepsilon_n(\theta)$, a.e.\,$\theta$.

Now, consider the self-adjoint operator
\[Q^\varepsilon_n:=\int^\oplus_{\mathcal{C}} -\Delta_n^\varepsilon(\theta)\,d\theta,\]
where
\[\dom Q^\varepsilon_n:=\left\{\varphi:\varphi(\theta)\in \dom (-\Delta_n^\varepsilon(\theta)), \text{a.e.} \, \theta;
\int_{\cal C} \| -\Delta_n^\varepsilon(\theta) \varphi(\theta) \|_{{\cal H}_n^L}^2 {\rm d} \theta < +\infty\right\}.\]
	
For each  $\varphi \in \dom q^\varepsilon_n$ and for each $\eta \in \dom Q^\varepsilon_n$, 
\begin{eqnarray*}
q^\varepsilon_n(\varphi,\eta)&=&\int_{\mathcal{C}}\hat{b}^\varepsilon_n(\theta)\left(\varphi(\theta),\eta(\theta)\right)\,d\theta\\
&=& \int_{\mathcal{C}}\left\langle\varphi(\theta), -\Delta_n^\varepsilon(\theta)\eta(\theta)\right\rangle_{{\cal H}^L_n}d\theta\\
&=& \int_{\mathcal{C}}\left\langle\varphi(\theta),(Q^\varepsilon_n\eta)(\theta)\right\rangle_{{\cal H}^L_n}d\theta\\
&=& \left\langle \varphi, Q^\varepsilon_n \eta \right\rangle.
\end{eqnarray*} 	
Therefore, $Q^\varepsilon_n$ is the self-adjoint operator associated with $q^\varepsilon_n$ and,
by uniqueness, 
$Q^\varepsilon_n=\mathcal{U}_n(-\Delta_n^\varepsilon)\mathcal{U}_n^{-1}$. 
\end{proof}

\subsection{Proof of Theorem \ref{theoremperiodiccase}}\label{prooftheosub}

Since each  $-\Delta_n^\varepsilon(\theta)$ has compact resolvent and is lower bounded, 
its spectrum is discrete. 
We denote by $E_{n,j}(\varepsilon,\theta)$ the $j$th eigenvalue of $-\Delta_n^\varepsilon(\theta)$, counted with multiplicity, and  by 
$\psi_{n,j}(\varepsilon,\theta)$ the corresponding normalized eigenfunction, i.e., 
\[-\Delta_n^\varepsilon(\theta)\psi_{n,j}(\varepsilon,\theta)=E_{n,j}(\varepsilon,\theta)\psi_{n,j}(\varepsilon,\theta), \quad j=1,2,3,\cdots, \quad \theta\in \mathcal{C}.\]   
	  
We have 
\[E_{n,1}(\varepsilon,\theta)\leq E_{n,2}(\varepsilon,\theta)\leq\cdots\leq E_{n,j}(\varepsilon,\theta)
\leq\cdots, \quad\theta \in \mathcal{C},\]
\[\sigma(-\Delta_n^\varepsilon)=\cup_{j=1}^\infty \{E_{n,j}(\varepsilon,\mathcal{C})\}, 
\quad \text{where}\quad E_{n,j}(\varepsilon,\mathcal{C}):=\{E_{n,j}(\varepsilon,\theta):\theta\in \mathcal{C}\};\]
each $E_{n,j}(\varepsilon,\mathcal{C})$ is called of the $j$th band of $\sigma(-\Delta_n^\varepsilon)$.

Lemma \ref{analytic}  ensures that the functions $E_{n,j}(\varepsilon,\theta)$ are real analytic functions in $\theta$;  
consequently, each $E_{n,j}(\varepsilon,\mathcal{C})$ is either a closed interval or a one point set.  
The goal is to find an asymptotic behavior for the eigenvalues $E_{n,j}(\varepsilon, \theta)$, as $\varepsilon \to 0$.

Based on the discussion in the Introduction, we start to study the sequence
\begin{equation}\label{beasympeigenuni}
b^\varepsilon_n(\theta)(\psi):=\hat{b}^\varepsilon_n(\theta)(\psi)-
\frac{\lambda_n}{\varepsilon^2}\|\psi\|_{{\cal H}_n^L}^2,
\end{equation}
$\dom  b^\varepsilon_n(\theta):= \dom \hat{b}^\varepsilon_n(\theta)$. 
The self-adjoint operator associated with
$b^\varepsilon_n(\theta)$
is  $T^\varepsilon_n(\theta):= -\Delta_n^\varepsilon(\theta)-(\lambda_n/\varepsilon^2) \Id$.

Define the one-dimensional quadratic form
 \begin{eqnarray*}\label{fibers-quadratics-forms-limit}
 b_n(\theta)(w)&:=& b^\varepsilon_n(\theta)(wu_n)\\
 &=& \int_{Q_L}|w' u_n + i \theta w u_n + \langle\nabla_y u_n, Ry \rangle\alpha'(s)w|^2 \ds\dy,
 \end{eqnarray*}
 $\dom b_n(\theta):= \{w\in H^1(0,L): w(0)=w(L)\}$. 
Denote by $T_n(\theta)$ the self-adjoint operator associated with it. Namely, 
\[T_n(\theta)w:=(-i\partial_s+\theta)^2w+ V_n w,\]
$\dom T_n(\theta)=\{w\in H^2(0,L): w(0)=w(L), w'(0)=w'(L)\}$, where $V_n$ is defined by (\ref{effectivepotential})
in the Introduction.
We have

\begin{Theorem}\label{converges-norm-reslvents-fibers}
For  each $n \in \mathbb N$ and each $\theta \in {\cal C}$ fixed, 
the sequence of self-adjoint operators $T^\varepsilon_n(\theta)$ 
converges in the norm resolvent sense 
to $T_n(\theta)$ in $\mathcal{L}_n^L$, as $\varepsilon\rightarrow 0$. 
Furthermore, for $n \in \mathbb N$,
$j \in \mathbb{N}$ and $\theta \in {\cal C}$ fixed, one has 
\[\lim\limits_{\varepsilon\rightarrow 0}\left(E_{n,j}(\varepsilon, \theta)-\frac{\lambda_n}{\varepsilon^2}\right)= k_{n,j}(\theta).\]
\end{Theorem}

The proof  of Theorem \ref{converges-norm-reslvents-fibers} is very similar to the 
proof of Theorem \ref{maintheorem}; it will be omitted here.

Denote by $k_{n,j}(\theta) $ the $j$th eigenvalue (counted multiplicity) of $T_n(\theta)$.
As a consequence  of Theorem \ref{converges-norm-reslvents-fibers}, we have 

\begin{Corollary}\label{corbasynw}
For each $n \in \mathbb N $ and each $j \in \mathbb{N}$ fixed, one has 
\begin{equation}\label{convergens-eigenvalue}
\lim\limits_{\varepsilon\rightarrow 0}\left(E_{n,j}(\varepsilon, \theta)-\frac{\lambda_n}{\varepsilon^2}\right)= k_{n,j}(\theta), 
\end{equation}
 uniformly in ${\cal C}$.   
\end{Corollary}	
\begin{proof}
For $n \in \mathbb N$ fixed, extend $b_n^\varepsilon(\theta)$ by the formulas (\ref{fibersquaform}) and (\ref{beasympeigenuni}), for all $\theta \in \overline{{\cal C}}$.
Theorem \ref{converges-norm-reslvents-fibers} holds true if we consider $\overline{{\cal C}}$ instead of ${\cal C}$.
Then, (\ref{convergens-eigenvalue}) holds true for each $j \in \mathbb N$  and each $\theta \in \overline{{\cal C}}$. 
On the other hand, if  $\varepsilon_1<\varepsilon_2$, then
$b^{\varepsilon_2}_n(\theta)(\psi) \leq b^{\varepsilon_1}_n(\theta)(\psi)$, for all $\psi \in \dom b^\varepsilon_n(\theta)$, 
for all $\theta\in \overline{{\cal C}}$.  
Thus, for each $j \in \mathbb N$ and each $\theta \in \overline{{\cal C}}$, the sequence
$(E_{n,j}(\varepsilon,\theta)-\lambda_n/\varepsilon^2)$ is decreasing in $\varepsilon$.
Now, the result follows by  Dini's Theorem.
\end{proof}

\vspace{0.3cm}
\noindent
{\bf Proof of Theorem \ref{theoremperiodiccase}:}
Let $E>0$, without loss of generality, we can suppose that, for all $\theta\in \mathcal{C}$, 
the spectrum of $-\Delta_n^\varepsilon(\theta)$ below $E$ consists of exactly $j_0$  eigenvalues 
$\{E_{n,j}(\varepsilon,\theta)\}_{j=1}^{j_0}$. 
Lemma \ref{analytic} ensures that $E_{n,j}(\varepsilon,\theta)$ and $\psi_{n,j}(\varepsilon, \theta)$ are real analytic functions
in $\theta \in {\cal C}$.

Theorem XIII in \cite{reedsimon4} implies that the functions $k_{n,j}(\theta)$ are  nonconstant. 
By Corollary \ref{corbasynw}, there exist $\varepsilon_E > 0$, $K(\varepsilon) > 0$, so that, 
$|E_{n,j}(\varepsilon, \theta)- (\lambda_n/\varepsilon^2)-\kappa_{n,j}(\theta)|
< K(\varepsilon)$, for all $\theta \in {\cal C}$, for all $\varepsilon \in (0, \varepsilon_E)$,
for all $j = 1, 2, \cdots, j_0$, and $K(\varepsilon) \to 0$, as $\varepsilon \to 0$.
Consequently, the functions $E_{n,j}(\varepsilon, \theta)$ are nonconnstant.
Note that $\varepsilon_E > 0$ depends on $j_0$, i.e., the thickness of the
tube depends on the length of the energies to be covered.
Now, by Section XIII.16 in \cite{reedsimon4}, the
conclusion follows.

\subsection{Existence of band gaps}\label{subsection-gaps}

In this section we are going to prove Theorem \ref{gap-exists}.
Consider the one-dimensional operator
\[\tilde{T}_nw:=-w''+V_nw,\quad \dom \tilde{T}_n=H^2(\mathbb{R}).\]
	
We have denoted by  $k_{n,j}(\theta)$ the $j$th eigenvalue (counted with multiplicity) of the operator 
$T_n(\theta)$. For each $j\in \mathbb{N}$, $k_{n,j}(\theta)$ is a real analytic function in $\mathcal{C}$. 
By Chapter XIII.16 in \cite{reedsimon4}, we have the following properties:

\vspace{0.3cm}
\noindent
(a) $k_{n,j}(\theta)=k_{n,j}(-\theta)$ for all $\theta\in \mathcal{C}$, $j=1,2,3,\cdots.$

\vspace{0.3cm}
\noindent
(b) For $j$ odd (resp. even), $ k_{n,j}(\theta)$ is strictly monotone increasing (resp. decreasing) as $\theta$ increases from $0$ to $\pi/L$. 
In particular, 
\[k_{n,1}(0)<k_{n,1}(\pi/L)\leq k_{n,2}(\pi/L)<k_{n,2}(0)  \leq \cdots  \leq k_{n,2j-1}(0)<k_{n,2j-1}(\pi/L)\]
\[\leq k_{n,2j}(\pi/L)<k_{n,2j}(0) \leq \cdots.\]

For each $j \in \mathbb N$, define
\begin{displaymath}
B_{n,j}:=\left\{ \begin{array}{lcc}
[ k_{n,j}(0),k_{n,j}(\pi/L) ], & \text{for}\,\, j\,\, \text{odd},\\
\left[k_{n,j}(\pi/L),k_{n,j}(0) \right], & \text{for}\,\, j\,\, \text{even},
\end{array}\right.
\end{displaymath}
and
\begin{displaymath}
G_{n,j}:=\left\{ \begin{array}{ll}
( k_{n,j}(\pi/L),k_{n,j+1}(\pi/L) ), & \text{for}\,\, j\,\, \text{odd so that } k_{n,j}(\pi/L)\neq k_{n,j+1}(\pi/L),\\
(k_{n,j}(0),k_{n,j+1}(0) ), & \text{for}\,\, j\,\, \text{even so that } k_{n,j}(0)\neq k_{n,j+1}(0) ,\\ 
\emptyset, & \text{otherwise}.
 \end{array}\right.
 \end{displaymath}	

Then, by Theorem XIII.90 in \cite{reedsimon4}, one has $\sigma(\tilde{T}_n)= \cup_{j=1}^\infty B_{n,j}$, 
where $B_{n,j}$ is called of the $j$th band of $\sigma(\tilde{T}_n) $. If $G_{n,j}\neq \emptyset$,   $G_{n,j}$ is called of
gap of $\sigma(\tilde{T}_n)$.  

By Corollary \ref{corbasynw} and since $E_{n,j}(\varepsilon, \theta)$ is a decreasing sequence,
for each $j \in \mathbb{N}$,  and for each $\varepsilon>0$,
\[\max\limits_{\theta\in\mathcal{C}}E_{n,j}(\varepsilon,\theta) \leq \left\{ \begin{array}{ll}
\lambda_n/\varepsilon^2 +k_{n,j}(\pi/L), &\text{for}\,\, j\,\,\text{odd},\\
\lambda_n/\varepsilon^2+ k_{n,j}(0), & \text{for}\,\, j\,\, \text{even},
\end{array}\right..\]
If $G_{n,j}\neq \emptyset$, again by Corollary \ref{corbasynw}, there exists $\varepsilon_j>0$, so that, 
for all $\varepsilon\in (0,\varepsilon_j)$,   
\[\min\limits_{\theta\in\mathcal{C}}E_{n,j+1}(\varepsilon,\theta) \geq \left\{ \begin{array}{ll}
\lambda_n/\varepsilon^2 +k_{n,j+1}(\pi/L)-|G_{n,j}|/2, &\text{for}\,\, j\,\,\text{odd},\\
\lambda_n/\varepsilon^2+ k_{n,j+1}(0)-|G_{n,j}|/2, & \text{for}\,\, j\,\, \text{even},
\end{array}\right. \]
where $|\cdot|$ denotes the Lebesgue measure.  Thus, we have

\begin{Corollary}\label{corollary-gaps}
If $G_{n,j}\neq \emptyset$, there exists $\varepsilon_j>0$, so that, for all $\varepsilon\in (0,\varepsilon_j)$, 
\[\min\limits_{\theta\in\mathcal{C}}E_{n,j+1}(\varepsilon,\theta)-\max\limits_{\theta\in\mathcal{C}}E_{n,j}(\varepsilon,\theta)\geq \frac{1}{2}|G_{n,j}|.\]
\end{Corollary}
	
Another important tool to prove Theorem \ref{gap-exists} is the following result due to Borg  \cite{borg}.
 
\begin{Theorem}\label{borg}
(Borg)
Suppose that $W$ is a real-valued, piecewise continuous function on $[0,L]$.  Let $\mu^{\pm}_j$ be the $j$th eigenvalue of the following operator counted multiplicity respectively 
\[ T^\pm:=-\frac{d^2}{ds^2}+ W(s), \quad \text{in}\quad L^2(0,L),\]
with domain
\[\{w\in H^2(0,L): w(0)=\pm w(L), w'(0)=\pm w'(L)\}.\]
We suppose that 
 \[ \mu^{+}_j=\mu^+_{j+1}, \quad\text{for all even }\, j,\]
 and 
 \[ \mu^{-}_j=\mu^-_{j+1}, \quad\text{for all odd }\, j.\]
Then, W is constant on $[0,L]$.
\end{Theorem}

\vspace{0.3cm}
\noindent
{\bf Proof of Theorem \ref{gap-exists}:}
Take $W(s)=V_n(s)$ in Theorem \ref{borg}. 
The operator $T_n(0)$ (resp. $T_n(\pi/L)$ ) is unitarily equivalent to $T^+$ 
(resp. $T^-$); in fact, just to consider the unitary operator $(u_\theta w)(s):=e^{-i\theta s} w(s)$ with $\theta =0$ 
(resp. $\theta = \pi/L$).
Remember that $\{k_{n,j}(0)\}_{j \in \mathbb N}$ (resp. $\{k_{n,j}(\pi/L)\}_{j \in \mathbb N}$) are the eigenvalues of $T_n(0)$ 
(resp. $T_n(\pi/L)$).

Since $V_n(s)$ is not a constant function in $[0,L]$, by Borg's Theorem, without loss of generality, we can affirm that there exists    $j\in \mathbb{N}$ so that  $k_{n,j}(0)\neq k_{n,j+1}(0)$. Now,  the result follows by Corollary \ref{corollary-gaps}.

\subsection{Location of band gaps}\label{locationbandsec}

In this section we find a location in $\sigma(-\Delta_n^\varepsilon)$ where Theorem \ref{gap-exists} holds true. 
For this purpose, we use the scaling
\begin{equation}\label{scaling}\alpha\mapsto \gamma \alpha,
\end{equation} 
where $\gamma>0$ is a small parameter.
Thus, we obtain the waveguide $\Omega^\alpha_{\varepsilon,\gamma}:=\Omega^{\gamma\alpha}_{\varepsilon}$.
Consider $-\Delta^N _{\Omega^\alpha_{\varepsilon,\gamma}}$  instead of $-\Delta^N_{\Omega^\alpha_{\varepsilon}}$ in the Introduction.
Denote by  $\overline{b}^{\varepsilon,\gamma}_n $ and $\hat{b}^{\varepsilon,\gamma}_n(\theta)$ the quadratic forms obtained by replacing 
(\ref{scaling}) in (\ref{quadrestintro}) and (\ref{fibersquaform}), respectively. 
The self-adjoint operators associated with these quadratic forms are denoted by 
$-\Delta^{\varepsilon,\gamma}_n$ and $-\Delta^{\varepsilon,\gamma}_n(\theta)$, respectively.
Denote  by $E_{n,j}{(\gamma,\varepsilon,\theta)} $ the $j$th eigenvalue of 
$-\Delta^{\varepsilon,\gamma}_n(\theta)$ counted with multiplicity.

Define $W_n(s):=C_n^1(S)(\alpha'(s))^2.$ Write $W_n(s)$ as 
a Fourier Series, i.e., 
\[W_n(s)=\sum_{j=-\infty}^{+\infty}\frac{1}{\sqrt{L}} w_n^j e^{2\pi jis/L} \quad \text{in}\quad L^2[0,L].\]
The sequence $\{w_n^j\}_{j=-\infty}^{\infty}$ is called of Fourier coefficients of $W_n$. 
Since $W_n$ is a real function, $ w_n^j=\overline{w_n^{-j}}$, for all $j\in \mathbb{Z}$.
We have

\begin{Theorem}\label{location}
Suppose that $V_n(s)$ is not a constant function in $[0,L]$ and $W_n(s)$ is non null. 
Let $j\in \mathbb{N}$ so that $w_n^j\neq 0$. Then, there exist $\gamma> 0$ small enough, $\varepsilon_{n,j+1}>0$
and $C_{n,j}(\gamma)>0$, so that, for all $\varepsilon\in(0,\varepsilon_{n,j+1})$,
\[\min_{\theta\in \mathcal{C}}E_{n,j+1}(\gamma,\varepsilon,\theta)-\max_{\theta\in \mathcal{C}} 
E_{n,j}(\gamma,\varepsilon,\theta)\geq C_{n,j}(\gamma).\]
\end{Theorem}

To prove Theorem \ref{location} we are going to use a strategy adopted in \cite{yoshitomi}. Some steps will be omitted here
and a more complete proof can be found in that work. In addition, our problem requires some more adjustments
which will be  explained in the next paragraphs.

\vspace{0.3cm}
\noindent
{\bf Some technical details.}
Let $W\in L^2[0,L]$ be a real function. For $\beta \in \mathbb{C}$, consider the operators 
\[T_\beta^+w=-w''+\beta W(s)w,\quad  \text{and} \quad T_\beta^-w=-w''+\beta W(s)w, \]
with domains given by 
\begin{equation}\label{condperi}
\dom T_\beta^+ = \{w\in H^2(0,L): w(0)=w(L), w'(0)= w'(L)\},
\end{equation}
\begin{equation}\label{condantiperi}
\dom T_\beta^-  =  \{w\in H^2(0,L): w(0)= - w(L), w'(0)=- w'(L)\},
\end{equation}
respectively. 
Denote by  $\{l^+_j(\beta)\}_{j\in \mathbb{N}}$ and $\{l^-_j(\beta)\}_{j\in\mathbb{N}} $ the eigenvalues of 
$T_\beta^+$ and $T_\beta^-$, respectively. For $\beta\in \mathbb{R}$ and $j\in \mathbb{N}$, define 

\[\delta_j^+(\beta):= l^+_{2j+1}(\beta)-l^+_{2j}(\beta) \quad\text{and} \quad \delta^-_j(\beta):=l^-_{2j}(\beta)-l^-_{2j-1}(\beta).\]
Now, \[ \delta_{2j-1}(\beta):=\delta^-_j (\beta) \quad\text{and} \quad \delta_{2j}(\beta):=\delta^+_j (\beta).\]

Let  $\{w^j\}_{j=-\infty}^{+\infty}$ be the Fourier coefficients of $W$;
\[W(s)=\sum_{j=-\infty}^{+\infty}\frac{1}{\sqrt{L}} w^j e^{2\pi jis/L} \quad \text{in}\quad L^2[0,L],\]
where $ w^j=\overline{w^{-j}}$, for all $j\in \mathbb{Z}$.

The next theorem gives an asymptotic behavior for $\delta_j(\beta)$, as $\beta\rightarrow 0$, in terms of the Fourier coefficients of $W$. 

\begin{Theorem}\label{asymptotic-yositomi}
For each $j \in \mathbb N$, 
\[\delta_j(\beta)=\frac{2}{\sqrt{L}}|w^j|\,|\beta|+O(|\beta|^2),\quad \beta\rightarrow 0, \, \beta\in \mathbb{R}. \]
\end{Theorem}

A detailed proof of Theorem \ref{asymptotic-yositomi} can be found in \cite{yoshitomi}.

\vspace{0.3cm}
\noindent
{\bf Auxiliary problem.}
For each $\gamma>0$ and $\theta\in \mathcal{C}$, consider the one-dimensional quadratic form
\[s^{\gamma}_n(\theta)(w):= \int^L_0\left(|w'+i\theta w|^2+\gamma^2 W_n(s)|w|^2\right) \ds,\]
$\dom s^\gamma_n(\theta):=\{w\in H^1(0,L): w(0)=w(L)\}$. The self-adjoint operator associated with $s^{\gamma}_n(\theta)$ is
given by 
\[S_n^\gamma(\theta)w:= (-i\partial_s+\theta)^2w+\gamma^2W_n(s)w,\]
$\dom S_n^\gamma(\theta):=\{w\in H^2(0,L): w(0)=w(L),w'(0)=w'(L)\}$.  Denote by 
$\nu_{n,j}(\gamma,\theta)$ the $j$th eigenvalue of $S^\gamma_n(\theta)$ counted with multiplicity.

Now, consider
\[b^{\gamma}_n(\theta)(w):= b^{\varepsilon, \gamma}_n(\theta)(wu_n) = 
\int^L_0\left(|w'+i\theta w|^2+ V_n^\gamma(s)|w|^2\right) \ds,\]
$\dom b^\gamma_n(\theta):=\{w\in H^1(0,L): w(0)=w(L)\}$, where
$V_n^\gamma(s):=\gamma^2 W_n(s) -\gamma C^2_n(S)\alpha''(s)$.
The self-adjoint operator associated with $b^{\gamma}_n(\theta)$ is 
\[T_n^\gamma(\theta)w:= (-i\partial_s+\theta)^2w+V^\gamma_n(s)w,\]
$\dom T_n^\gamma(\theta):=\{w\in H^2(0,L): w(0)=w(L),w'(0)=w'(L)\}$.  
Denote by $k_{n,j}(\gamma,\theta)$ the $j$th eigenvalue of $T^\gamma_n(\theta)$ counted with multiplicity.  
  
Take $c>\max\{\|V_n\|_{\infty}, \|W_n\|_{\infty}\}$.
Some straightforward calculations show that there exists $K>0$, so that, 
\begin{equation}\label{inequadformusefulappox}
|\left(b_n^\gamma(\theta)+c\right)(w)-\left(s^\gamma_n(\theta)+c\right)(w)|\leq 
K \, \gamma \, |\left(b_n^\gamma(\theta)+c\right)(w)|,\quad \forall w\in \dom b^\gamma_n(\theta),
\end{equation}
for all $\theta\in \mathcal{C}$, for all $\gamma>0$ small enough.

Inequality (\ref{inequadformusefulappox}), Theorem 2 in \cite{bdv},  and Corollary 2.3 in \cite{gohberg} imply

\begin{Corollary}\label{corollary-gamma-limit}
For each $j\in \mathbb{N}$, there exists $\gamma_j>0$, so that, for all $\gamma\in(0,\gamma_j)$, 
\[k_{n,j}(\gamma,\theta)=\nu_{n,j}(\gamma,\theta)+O(\gamma),\]
uniformly in $\mathcal{C}$. 
\end{Corollary}

\vspace{0.3cm}
\noindent
{\bf Some estimates.} {\bf I.}
We define
\begin{displaymath}
G_{n,j}(\gamma):=\left\{ \begin{array}{ll}
(k_{n,j}(\gamma,\pi/L),k_{n,j+1}(\gamma,\pi/L) ), & \text{for}\,\, j\,\, \text{odd so that } 
k_{n,j}(\gamma,\pi/L)\neq k_{n,j+1}(\gamma,\pi/L),\\
(k_{n,j}(\gamma,0),k_{n,j+1}(\gamma,0) ), & \text{for}\,\, j\,\, \text{even so that } k_{n,j}(\gamma,0)\neq k_{n,j+1}(\gamma,0) ,\\ 
\emptyset, & \text{otherwise}.
\end{array}\right..
\end{displaymath}	
Namely, if $G_{n,j}(\gamma) \neq \emptyset$,  it is called of gap of the spectrum $\sigma(T_n^\gamma)$, where
\[T_n^\gamma w := -w''+V_n^\gamma(s)w, \quad \dom T_n^\gamma = H^2(\mathbb R).\]

Similarly to the considerations of Section \ref{subsection-gaps} and Corollary \ref{corollary-gaps}, we have
\begin{Corollary}\label{corollary-gaps-location}
If $G_{n,j}(\gamma)\neq \emptyset$, there exist $\gamma_j > 0$ and $\varepsilon_j>0$, so that, for all 
$\gamma \in (0,\gamma_j)$ and for all $\varepsilon\in (0,\varepsilon_j)$, 
\[\min\limits_{\theta\in\mathcal{C}}E_{n,j+1}(\gamma,\varepsilon,\theta)-\max\limits_{\theta\in\mathcal{C}}
E_{n,j}(\gamma,\varepsilon,\theta)\geq\frac{1}{2}|G_{n,j}(\gamma)|.\]	
\end{Corollary}
	
\vspace{0.3cm}
\noindent
{\bf II.} Now, we consider
\begin{displaymath}
\tilde{G}_{n,j}(\gamma):=\left\{ \begin{array}{ll}
( \nu_{n,j}(\gamma,\pi/L),\nu_{n,j+1}(\gamma,\pi/L) ), & \text{for}\,\, j\,\, \text{odd so that } 
\nu_{n,j}(\gamma,\pi/L)\neq \nu_{n,j+1}(\gamma,\pi/L),\\
(\nu_{n,j}(\gamma,0),\nu_{n,j+1}(\gamma,0) ), & \text{for}\,\, j\,\, \text{even so that } \nu_{n,j}(\gamma,0)\neq \nu_{n,j+1}(\gamma,0) ,\\ 
\emptyset, & \text{otherwise}.
\end{array}\right.;
\end{displaymath}	
if $\tilde{G}_{n,j}(\gamma) \neq \emptyset$, it is called of gap of $\sigma(S_n^\gamma)$, where
\[S_n^\gamma w := -w''+ \gamma^2 W_n(s)w, \quad \dom S_n^\gamma = H^2(\mathbb R).\]

As in the proof of Theorem \ref{gap-exists},
consider the unitary operator $(u_\theta w)(s)=e^{-i\theta s} w(s)$.
We define the self-adjoint operators $\tilde{S}_n^\gamma (0) := u_0 S_n^\gamma(0) u_0^{-1}$ and 
$\tilde{S}_n^\gamma(\pi/L) := u_{\pi/L} S_n^\gamma (\pi/L) u_{\pi/L}^{-1}$, whose
eigenvalues are given by $\{\nu_{n,j}(\gamma, 0)\}_{j \in \mathbb N}$ and $\{\nu_{n,j}(\gamma, \pi/L)\}_{j \in \mathbb N}$,
respectively.
Furthermore,
the domains of these operators are given by (\ref{condperi}) and (\ref{condantiperi}), respectively.
Thus,
we can see that $|\tilde{G}_{n,j}(\gamma)|=\delta_j(\gamma^2)$, for all $j \in \mathbb{N}$, if we consider $\beta
=\gamma^2$ and $W(s)=W_n(s)$ in Theorem \ref{asymptotic-yositomi}.

\vspace{0.3cm}
With the notes of the previous paragraphs, we have conditions to prove the main theorem of this section.

\vspace{0.3cm}
\noindent
{\bf Proof of Theorem \ref{location}:}
Recall that we have denoted by $\{w^j_n\}_{j=-\infty}^{+\infty}$ the Fourier coefficients of $W_n$. 
Since $W_n$ isn't a constant function in $[0,L]$, there exists $j\in\mathbb{N}$, so that, $w^{j}_n\neq 0$.   
By Theorem \ref{asymptotic-yositomi}, 
\[|\tilde{G}_{n,j}(\gamma)|= \frac{2}{\sqrt{L}} \gamma^2|w^{j}_n|+O(\gamma^4), \quad\gamma\rightarrow 0.\]
This estimate and Corollary \ref{corollary-gamma-limit} imply that $|G_{n,j}(\gamma)|>0$, for all
$\gamma>0$ small enough. By Corollary \ref{corollary-gaps-location}, 
theorem is proven by taking  $C_{n,j}(\gamma):=|G_{n,j}(\gamma)|/2>0$.

\begin{Remark}{\rm
Since we suppose that $V_n(s)$ is a non null function in $[0,L]$,
if $W_n(s) = 0$, for all $s \in \mathbb R$, one can consider $\tilde{W}_n(s) := C_n^2(S) \alpha''(s)$ instead of $W_n(s)$
in this subsection.
All the previous results also hold true in this case; the proofs are similar and will not be presented here.}
\end{Remark}

\appendix

\section{Appendix}

\subsection{The self-adjoint operator associated with $b_n$}\label{app001}

Recall the quadratic form 
\[b_n(w) =\int_Q|w'u_n+\langle\nabla_yu_n,Ry\rangle \alpha'(s)w|^2 \ds,\]
$\dom b_n =H^1(I)$.
The goal is to show that the operator $T_n$ defined by
(\ref{effectivepotential02}), (\ref{effectivepotential}), (\ref{effectivepotential01}) and (\ref{domintrorobin})
 in the Introduction is the self-adjoint operator associated with $b_n$.

Consider the particular case  where $I=(a,b)$ is a bounded interval.
Some calculations show that 
\[b_n(w) =\int_a^b\left(|w'|^2+V_n(s)|w|^2\right) \ds + C_n^2(S) \alpha'(b) |w(b)|^2 - C_n^2(S) \alpha'(a) |w(a)|^2.\] 
Let $b_n(w,u)$ be the sesquilinear form associated with $b_n(w)$.
We have
\[b_n(w,u)=\langle w,T_nu\rangle,\quad \forall w\in \dom b_n, \forall v\in \dom T_n.\]
Then,
$T_n$ is self-adjoint operator associated with $b_n$.
The case $I=\mathbb R$ can be proven in a similar way.

\subsection{$\Gamma$-convergence}\label{app1}

Let $H$ be a (real or complex) Hilbert space and $\overline{\mathbb R}= \mathbb R \cup \{ +\infty\}$.
The sequence of quadratic functionals $f_\varepsilon: H \rightarrow \overline{\mathbb{R}}$ 
strongly $\Gamma$-converges to  $f:H \rightarrow \overline{\mathbb{R}}$ 
(that is, $f_\varepsilon\xrightarrow{S\Gamma} f$) iff the following two conditions are satisfied:

\vspace{0.3cm}
\noindent
$i)$
 For every $v\in H$ 
and every $v_\varepsilon \rightarrow v$ in $H$ one has 
\[\liminf\limits_{\varepsilon\rightarrow0}f_\varepsilon(v_\varepsilon) \geq f(v).\]

\noindent
$ii)$
For every $v\in H$, there exists a sequence $v_\varepsilon\rightarrow v$ in $H$ such that 
\[\lim\limits_{\varepsilon\rightarrow0}f_\varepsilon(v_\varepsilon)=f(v).\]

If the strong convergence $v_\varepsilon\rightarrow v$ is replaced by the weak convergence 
$v_\varepsilon\rightharpoonup v$ in  $i)$ and $ii)$, then one has a characterization of 
the weakly $\Gamma$-converge (i.e., $f_\varepsilon\xrightarrow{W\Gamma}f$).

The following result can be found in \cite{maso} where
is proven the version for real Hilbert spaces;
the generalization for complex Hilbert spaces is presented in \cite{cesargamma}.

\begin{Proposition}\label{propappfirst}
Let $d_\varepsilon; d$ be positive (or uniformly lower bounded) closed sesquilinear
forms in the Hilbert space $H$, and $D_\varepsilon; D$ the corresponding associated
positive self-adjoint operators. Then, the following statements are equivalent:

a) $d_\varepsilon \xrightarrow{S\Gamma}d$ and, for each $\zeta \in H$, $\displaystyle d(\zeta) \leq \liminf_{\varepsilon \to 0} 
d_\varepsilon(\zeta_\varepsilon)$, $\forall \zeta_\varepsilon \to \zeta$ in $H$.

b) $d_\varepsilon \xrightarrow{S\Gamma}d$ and $d_\varepsilon\xrightarrow{W\Gamma}d$.

c) $D_\varepsilon$ converges to $D$ in the strong resolvent sense in 
$H_0 = \overline{\dom D}\subset H$, that is,
\[\lim_{\varepsilon \to 0} R_{-\lambda}(D_\varepsilon) \zeta = R_{-\lambda}(D) P \zeta,
\quad \quad \forall \zeta \in H, \forall \lambda> 0,\]
where $P$ is the orthogonal projection onto $H_0$.
\end{Proposition}

The following result is due to \cite{cesargamma}.

\begin{Proposition}\label{propappsecond}
Let $d_\varepsilon, d \geq \beta > -\infty$ closed sesquilinear forms and
$D_\varepsilon, D \geq \beta \Id$ the corresponding associated self-adjoint operators, and let 
$\overline{\dom D} = H_0 \subset H$. Assume that the following three conditions hold:

a) $d_\varepsilon \xrightarrow{S\Gamma}d$ and $d_\varepsilon\xrightarrow{W\Gamma}d$.

b) The resolvent operator $R_{-\lambda}(D)$ is compact in $H_0$ for some real number
$\lambda > |\beta|$.

c) There exists a Hilbert space ${\cal K}$, compactly embedded in H, so that if
the sequence $(\psi_\varepsilon)$ is bounded in $H$ and $(d_\varepsilon(\psi_\varepsilon))$ is also bounded, then
$(\psi_\varepsilon)$ is a bounded subset of ${\cal K}$.

Then, $D_\varepsilon$ converges in norm resolvent sense to $D$ in $H_0$ as $\varepsilon \to 0$.
\end{Proposition}

\begin{Remark}{\rm
In both propositions, the
domain of $D$ is not supposed to be dense in $H$ but is required that
$\img D \subset H_0$; 
we say that $D$ is self-adjoint in $H_0$.}
\end{Remark}


%
%

\end{document}